\documentclass[graybox]{svmult}

\usepackage{type1cm}        % activate if the above 3 fonts are
\usepackage{makeidx}         % allows index generation
\usepackage{graphicx}        % standard LaTeX graphics tool
\usepackage{multicol}        % used for the two-column index
\usepackage[bottom]{footmisc}% places footnotes at page bottom

\usepackage{newtxtext}       % 
\usepackage{newtxmath}       % selects Times Roman as basic font

\makeindex             % used for the subject index
\usepackage{mathtools}
\usepackage{hyperref}
\usepackage{bbm}
\usepackage{algorithm} 
\usepackage{algpseudocode}
\usepackage{scrextend}

\usepackage{graphicx}
\graphicspath{ {./images/} }
\newcommand*{\p}{\mathbb{P}}
\usepackage{xcolor}
\usepackage{dsfont}
\usepackage{natbib}

\graphicspath{ {./images/} }
\usepackage{float}
\usepackage{multibib}
 \usepackage{algorithm} 
\newcites{SM}{References}

\usepackage{comment}

\newcommand{\norm}[1]{\left\lVert#1\right\rVert}
\usepackage{setspace}
\usepackage[capitalise]{cleveref}

\newtheorem{Lemma}{Lemma 1}
\usepackage[utf8]{inputenc}

\date{}

\definecolor{darkblue}{rgb}{.1, 0.1,.8}
\definecolor{darkgreen}{rgb}{0,0.8,0.2}
\definecolor{darkred}{rgb}{.8, .1,.1}
\definecolor{violet}{RGB}{148,0,211}

\newcommand*{\N}{\mathbb{N}}

\begin{document}
\title*{Uniform confidence bands for joint angles across different fatigue phases}
% Use \titlerunning{Short Title} for an abbreviated version of
% your contribution title if the original one is too long
\author{Patrick Bastian, Rupsa Basu, Holger Dette}
% Use \authorrunning{Short Title} for an abbreviated version of
% your contribution title if the original one is too long
\institute{Patrick Bastian \at Ruhr Universität Bochum, Universitätsstraße 150, 44801 Bochum, \email{patrick.bastian@rub.de}
\and Rupsa Basu \at Universität zu Köln, Universitätstraße 24, 50931 Köln, \email{rbasu@uni-koeln.de}
\and Holger Dette \at Ruhr Universität Bochum, Universitätsstraße 150, 44801 Bochum \email{holger.dette@rub.de}}
%
% Use the package "url.sty" to avoid
% problems with special characters
% used in your e-mail or web address
%
\index{Surname, N.} %First Author
\index{Surname, N.} %Second Author

\maketitle

\abstract{We develop uniform  confidence bands for the mean function of stationary time series as a post-hoc analysis of multiple change point detection in functional time series.
In particular, the methodology in this work provides bands for those segments where the jump size exceeds a certain threshold $\Delta$.
In \cite{bastian2024multiplechangepointdetection} such exceedences of $\Delta$ were related to fatigue states of a running athlete. The extension to confidence bands stems  from an interest in understanding the range of motion (ROM) of lower-extremity joints of  running athletes under fatiguing conditions.  From a biomechanical perspective, ROM serves as a proxy for joint flexibility under varying fatigue states, offering individualized insights into potentially problematic movement patterns. The new  methodology provides a valuable tool for understanding the dynamic behavior of joint motion and its relationship to fatigue.}
%From a biomechanical point of view, these  ROM  serve as a proxy for (joint-) flexibility under varying fatiguing conditions, providing an individualistic insight into potentially problematic movements. } %This study extends this work and stems from an interest in understanding the range of motion (ROM) of lower-extremity joints in running athletes under fatiguing conditions. In this case the ROM from individual athletes serves as a proxy for measuring the flexibility in joint movement when the athlete is undergoing fatigue. In this work we propose confidence bands which can be used as a tool to study the ROM.}

\section{Introduction}
Fatigue detection in biomechanical data of the lower extremity joint angles may be recast as change point detection in a functional time series, as seen previously in \cite{basu2023fatiguedetectionsequentialtesting} and \cite{bastian2024multiplechangepointdetection}.  Methods in these references  provide statistical methodologies for segmenting data based on the tiredness of athletes, ultimately providing phases where the athlete may be considered to be rested, pre-fatigued or fully exhausted. 
In this paper will explicate some of the consequences of theses results by studying the fatigue attributed impact on the range of motion (ROM) within these phases. To achieve this, we construct confidence bands for the mean curves corresponding to each fatigue state.

The dataset motivating our work comprises of measurements of lower-extremity joint angles collected from running athletes subjected to a stress-inducing protocol specifically designed to simulate fatiguing conditions. Under fatigue, the body is expected to make adjustments to maintain endurance, resulting in changes in joint angles. In  Figure \ref{fig:exampleFig01}, we provide a snapshot of knee angle data from a single runner. Running, being a repetitive activity, naturally exhibits cycles, also called strides, which are clearly visible. In the  upper part of the figure we see $8$ strides, but the full data set consists of more than $1200$ strides as shown in the lower part of Figure \ref{fig:exampleFig01}.   

\begin{figure}[t]
\hspace{-1.7cm}
    \includegraphics[width=1.25\linewidth, height = 4cm]{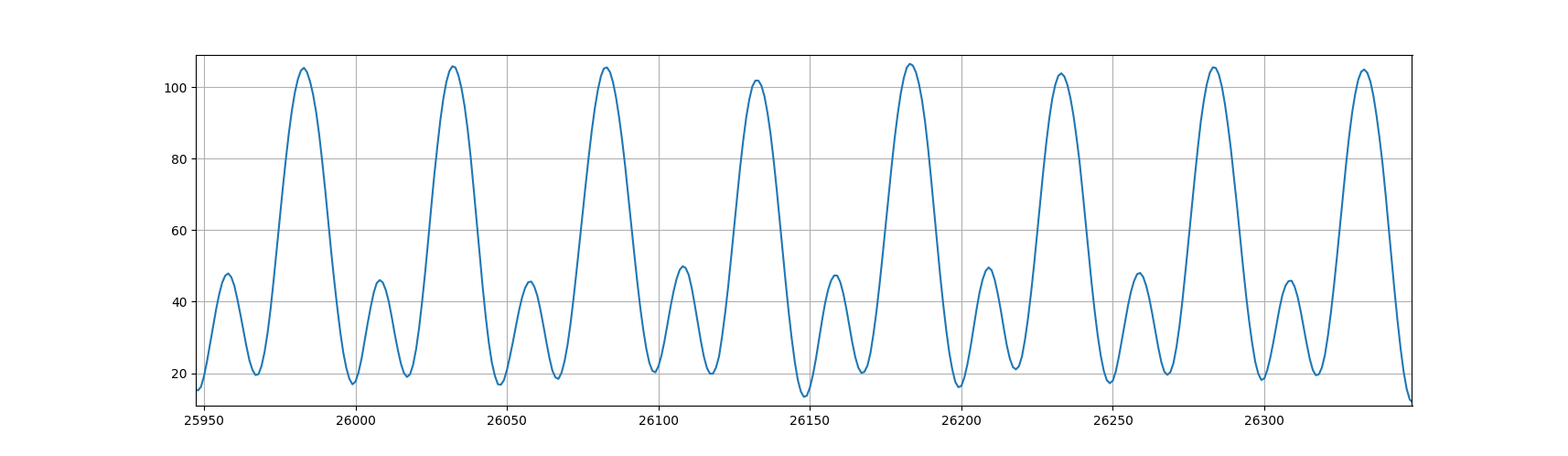}
    \vspace{-.5cm}
    \hspace{-0.21cm}
    \includegraphics[width = 1.0125\linewidth, height = 4cm]{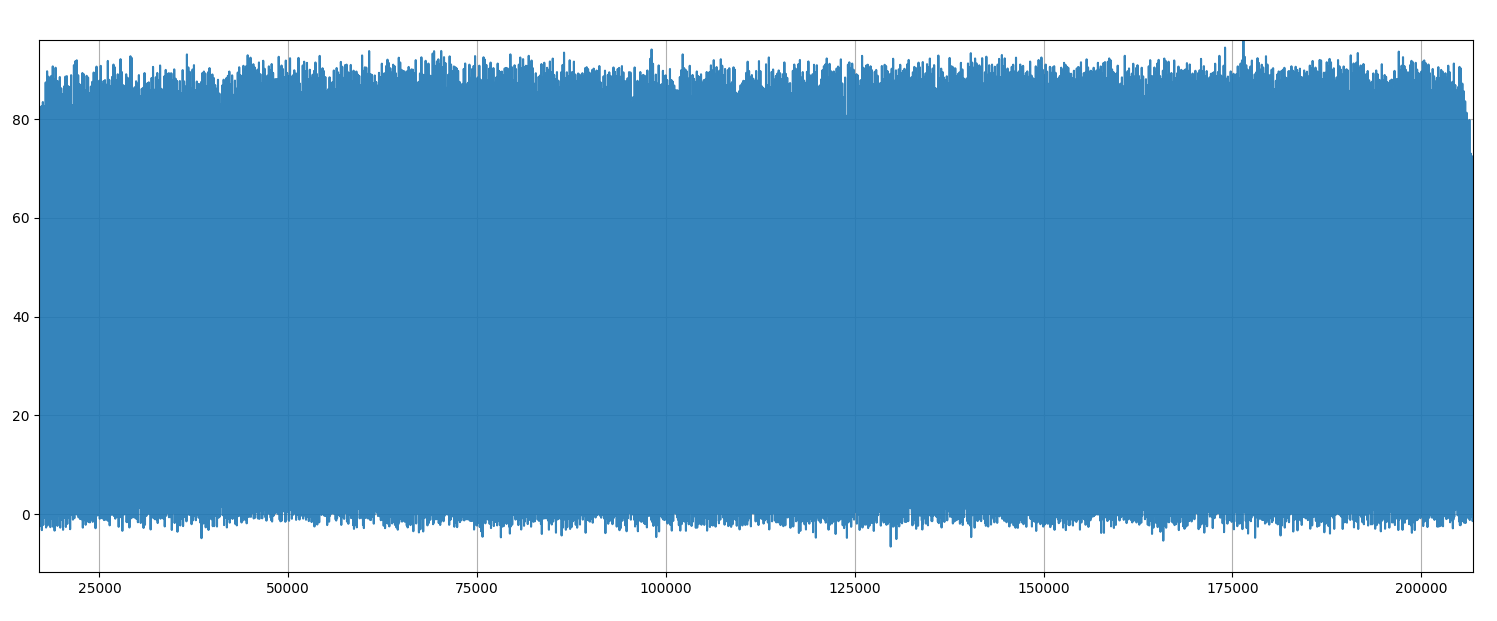}
    \vspace{-.1cm}
        \caption{Functional time series of  knee angle data from runners. The upper panel  shows $8$ cycles of the data starting and ending with the contact of the same foot with the ground. The lower panel shows the full functional time series of knee angle data from the above runner with about $1200$ cycles. Individual cycles not visible due to being condensed in the plot.}
    \label{fig:exampleFig01}
\end{figure}

Detecting fatigue related adaptations of the body can be considered as change detection problem for functional time series data. In this work, supposing that the locations of fatigue haven been identified, we address the problem of understanding the movement of the runner within the stable periods of no change by developing confidence bands for the mean function of the time series between two changes. 
Such bands are of interest to annotate the individual 
% range of motion 
ROM of the joint angles of a particular running athlete, thereby obtaining a better understanding of the ROM during rested, pre-fatigue and fatigued stages of the run. In particular they serve as a powerful visual tool for the biomechanical scientist to establish whether the athlete under study is somehow undergoing movement that is abnormal as compared to her standard course of motion.
 Ultimately, this knowledge is beneficial for ensuring that injuries resulting from harmful movements during fatigued phases is curtailed timely, thereby ensuring enhanced longevity of the runners  performance. 

%An example on the outcome of the methods presented in this paper is illustrated in \cref{fig:confiintro}, where  confidence bands from left to right for phases of increasing fatigue are presented.

%\begin{figure}[t]
%\hspace{-1.9cm}
%    \includegraphics[width=1.25\linewidth, height = 4.6cm]{images/intro_fig_knee.png}
%    \caption{Functional time series of  knee angle data from runners. Here we show $8$ cycles of the data starting and ending with the contact of the same foot with the ground.}
%    \label{fig:exampleFig01}
%    \vspace{0.5cm}
%    \hspace{-1.2cm}
%    \includegraphics[width=1.2\linewidth, height = 4.5cm]{images/intro_figconfi.png}
%    \caption{Confidence bands from various phases during course of the run from a single athlete from rest (\textbf{left}), pre-fatigued (\textbf{middle}) and fatigued (\textbf{right}) phases. The relevant size of change is $\Delta = 4$. }
%    \label{fig:confiintro}
%    %07_Rknee
%\end{figure}

\textbf{Summary of approach} %For the convenience of the reader we briefly recall the mathematical setup that facilitates the methodology developed in  \cite{bastian2024multiplechangepointdetection} and will also be used in this article. 
Consider that for a runner the lower extremity joint angle data from the knee \footnote{Similar analysis from the hip and ankle- angle data may be pursued. For sake of brevity, we omit the study of these examples in this paper. } obtained over a course of the run are  simply (discrete) realizations of an underlying function of time. Due to the cyclic nature of this data we have a long functional time series with each such function beginning and ending with  the subsequent contact of the same foot with the ground. Each of these cycles, or strides in biomechanical terms, may be represented mathematically as an observation taking values in the space of continuous functions, say $ C[0,1]$, leading to a time series, say $X_{n,1},...,X_{n,n} \in C[0,1]$, where $n$ is the total number of strides. We model these observations in the form, 
\begin{align}
    X_{n,j}(t)=\mu_{n,j}(t)+\epsilon_{n,j}(t) \quad j=1,...,n;
\end{align}
where $\{\epsilon_{n,j}\}_{j=1,\ldots n} $ is triangular array  whose rows are stationary and $\mu_{n,j}=\mathbb{E}[X_j]\in C[0,1]$ denotes the underlying mean function of $X_{n,j}$. Notice here that the index $j$  denotes the sequential nature of the cycles i.e., $j=1$ is the joint angle data from the first stride followed by $j=2$ and so on. Under this setup, detecting changes with statistical guarantees in the mean function $\mu_{n,j}$ then corresponds in biomechanical terms to detecting the onset of different fatigue states over the locations $j$ in the average movement of the runner. %Practitioners are interested in the range of motion of the joint angles in different states of fatigue.
In this work, we assume that $\mu_{n,j}$ is piecewise constant (in $j$) and that there are $m$ (rescaled) locations of change denoted by  $ 1 < \lfloor s_1 n \rfloor <  \ldots <  \lfloor s_m n \rfloor$, where $0 < s_1 < \ldots  < s_m <1$,  where we define $s_0 = 0 $ and $s_{m+1} = 1.$. More precisely, we assume that the mean function on the segment between
$j =\lfloor ns_{i} \rfloor$ and $j= \lfloor s_{i+1} n\rfloor$ 
is given by 
\begin{align}
    \mu_i(t):=\mu_{n,\lfloor s_{i}n \rfloor}(t) ~~, ~~i= \lfloor ns_{i} \rfloor \ldots  \lfloor s_{i+1} n\rfloor-1 .
\end{align}
 We are now interested in characterizing the range of motion in those intervals $[s_{i}, s_{i+1}]$ where the change at time $s_{i}$ corresponds to a transition from one fatigue state to another.

This poses three problems. First we need to estimate the change points in such a functional time series, secondly we need to determine which of these change locations are relevant, i.e. correspond to transitions of the fatigue state. Thirdly,  we then need to construct confidence bands for the associated segments. The first and second problems have been  investigated in  detail by  \cite{bastian2024multiplechangepointdetection}. In this note we will solve the third problem, i.e.  we construct upper and lower confidence bands, respectively given by functions $t \to \hat \mu_i^+(t),$ $ t \to \hat \mu_i^-(t)$ such that,

\begin{align}
 \label{p1}
   \lim_{n \to \infty}  \p\Big( \cap_{i \in I}\big\{
  \text{for all } t \in [0,1]:~ \hat \mu_i^-(t) \leq \mu_i(t) \leq \hat \mu_i^+(t) \big \}\Big)=1-\alpha, 
\end{align}
where  
\begin{align}
\label{det1}
   I \coloneqq \{0\}\cup \big \{1 \leq i \leq {m}  ~\big | ~ \norm{\mu_{i}-\mu_{i-1}}_\infty > \Delta \big \} 
\end{align}
is the set of changes  of size larger than $\Delta   >0$. Here 
we include $0$ for notational convenience of later statements.  $\Delta$ denotes a threshold that can be chosen either automatically or by the practitioner and it serves to eliminate subtle changes in the data which may be due to external conditions like pedestrians, obstacles, or curvature on the path of the athlete. The main contribution of this note is the construction of a bootstrap procedure that allows us to compute the lower and upper confidence functions $\hat \mu_i^\pm, i \in I$ that fulfill \eqref{p1}, provided  that  estimators of the change points $s_i$ and the set $I$ are available. 

With the additional contribution of this note we transform the biomechanical problem of studying the impact of fatigue on the range of motion into a three-step statistical problem of \textit{(i)} detecting locations of change $s_i$ in a functional time series, \textit{(ii)} demarcating those changes $s_i,  i \in I$ corresponding to relevant changes in the mean function, followed by \textit{(iii)} constructing simultaneous confidence bands for the mean functions  within the segments $[s_{i}, s_{i+1}], i \in I.$

\section{Confidence bands for multiple means}

Suppose there are a total of $m$ change locations  $s_1,...,s_m$ whereof we have $k \leq m$  locations of relevant change $s_{i_l}, l=1,...,k$ so that  the set $I$ in \eqref{det1} is given by $I=
%\{0\}\cup\{ 1 \leq i \leq m | \norm{\mu_i-\mu_{i-1}}> \Delta \}=
\{0,i_1,...,i_k\}$. We are interested in the $k+1$ segments of the functional time series given by $[0, s_{i_1}], [s_{i_1},s_{i_1+1}],..., [s_{i_k}, s_{i_k+1}]$ and their corresponding mean curves $\mu_0, \mu_{i_1},...,\mu_{i_k}$. Therefore, provided that the locations of the change  points $s_i, i=1, \dots, m$ and the indices $\{i_1,...,i_k\}$ are known or can be estimated, the problem in \eqref{p1} may be rephrased as obtaining confidence bands for the multivariate vector given by, 
\begin{align}
    \textbf{$\mu$}=(\mu_0, \mu_{i_1}, ..., \mu_{i_k})^\top \in (C[0,1])^{k+1}~.
\end{align}
The case $k=0$ corresponding to no change point has already been treated in \cite{Aue2020}. For the extension to the situation, where there are (relevant) change points, we require estimators for the the locations  $s_1,...,s_m$ of the changes, for their number $m$ and for the set $I$. For this purpose we denote by $\hat s_i, i=1,...,\hat m$ the estimators from Algorithm 1 of \cite{bastian2024multiplechangepointdetection}, where $\hat m$ is the estimated number of change points and additionally let $\hat I$ denote the estimator of the set $I$  of the relevant chnage points that can be obtained by Algorithm 2 in the same paper. Precise consistency statements for these estimates may be obtained in Theorem 4.1 and 4.2 
of  the given reference.

To obtain confidence bands we need to derive the asymptotic distribution of a suitable statistics such as
\begin{align}
\label{p2}
    T_n=\max_{i_l \in \hat I}\sqrt{\hat n_{i_l}}\norm{\frac{\hat \mu_{i_l}-\mu_{i_l}}{\hat \sigma}}_\infty~,
\end{align}
where $\hat n_i=\lfloor n(\hat s_{i+1}-\hat s_{i}) \rfloor$, $\hat \mu_i= n_i^{-1}\sum_{k=\lfloor n\hat s_{i}\rfloor}^{\lfloor n\hat s_{i+1}\rfloor-1}X_{n,k}$ and
$\hat \sigma^2$ is an estimator of the 
long run variance function
$$
\sigma^2 = \sum_{j\in \mathbb{Z}} \mathbb{E} [\epsilon_0\epsilon_j] ,
$$  which we define in equation \eqref{lrv} further below. One can show that it holds that $T_n \overset{d}{\rightarrow}T$ for some random variable $T$ and we can then obtain the required confidence bands by defining,
\begin{align}
    \mu^\pm_{i} (t)=\hat \mu_i (t) \pm \hat \sigma(t) q_{1-\alpha/2}/\sqrt{n}, \quad i \in \hat I,
\end{align}
where $q_{1-\alpha}$ is the $(1-\alpha)$-quantile of $T$. The distribution of $T$ however depends not only on the set $I$ but also on the covariance structure of the functional time series and as such needs to also be estimated from the data. For this purpose, we will employ a multiplier bootstrap procedure. 

We first construct a bootstrap version of $T_n$. To that end let $L=L(n)$ denote a block length and define the random variables
\begin{align}
    Y_{n,j}&=X_{n,j}-\hat \mu_i, \quad \quad \hat s_{i}\leq j/n < \hat s_{i+1}\\
    \hat \mu_i^*&=\hat n_i^{-1}\sum_{j=\lfloor n\hat s_{i}\rfloor}^{\lfloor n\hat s_{i+1}\rfloor-1}\nu_j\Big(L^{-1/2}\sum_{l=0}^{L-1} Y_{n,j+l}\Big)
\end{align}
where $\{ \nu_j\}_{j \in \N}$ is an iid sequence of standard normal random variables. The bootstrap test statistic is then given by
\begin{align}
    T_n^*=\max_{i \in \hat I}\sqrt{\hat n_i}\norm{\hat \mu_i^*/\hat \sigma}_\infty.
\end{align}
From the bootstrap procedure, we obtain statistics $T_{n}^{*, (1)}, \dots, T_{n}^{*, (R)}$ and denote the $(1-\alpha)$-quantile by $q_{1-\alpha}^*$. For convenience we summarize the above procedure in \cref{Alg1}. Finally we obtain the desired confidence bands by defining,
\begin{align}
\label{p3}
    \hat \mu^\pm_{i}=\hat \mu_i \pm \hat \sigma(t)\hat q^*_{1-\alpha/2}/\sqrt{\hat n_i}~, \quad i \in \hat I.
\end{align}

\begin{algorithm}[H]
	\begin{algorithmic}[1]
	\State \textbf{Compute} the number $\hat m$ and all (ordered) change points $\hat{S} = \{\hat{s}_1, \dots, \hat{s}_{\hat k} \}$ using BINSEG($1,n, \xi_n$) from \cite{bastian2024multiplechangepointdetection}
    \State \textbf{Compute} at significance level $\beta$ the set of relevant change points $\hat I$ by Algorithm 2 from \cite{bastian2024multiplechangepointdetection}
    \State \textbf{Compute} the means $\hat \mu_i= \hat n_i^{-1}\sum_{j=\lfloor n\hat s_{i}\rfloor}^{\lfloor n\hat s_{i+1}\rfloor-1}X_{n,j}$ for $i \in \hat I$
     \State \textbf{Compute} the long run variance $\hat \sigma^2$ as given in \eqref{lrv}
    \State \textbf{Compute} $Y_{n,j}=X_{n,j}-\hat \mu_i$ for  $\hat s_{i}\leq j/n < \hat s_{i+1}$
		\State \textbf{Fix} { block length $L$,  number of bootstrap replications $R$}  
         \For{ $r = 1, \dots, R$}
         \State \textbf{Compute}  $\hat \mu_i^*=\hat n_i^{-1}\sum_{j=\lfloor n\hat s_{i}\rfloor}^{\lfloor n\hat s_{i+1}\rfloor}\nu_j\Big(L^{-1/2}\sum_{l=0}^{L-1} Y_{n,j+l}\Big) $  for $i \in \hat I$ 
         \State \textbf{Compute}         
         $\hat T_{n}^* = \max_{i \in \hat I}\sqrt{\hat n_i}\norm{\hat \mu_i^*/\hat \sigma}_\infty$
         \EndFor	
	\State \textbf{Compute  }  $\hat q^*_{1-\alpha} \leftarrow $ 
 as the   empirical $(1-\alpha)$-quantile of bootstrap sample  $\hat T_{n}^{*,(1)} . \ldots , \hat T_{n}^{*, (R)}. $ 	
	\end{algorithmic}
 \caption{Bootstrapping Quantiles}
 \label{Alg1}
\end{algorithm}

\begin{theorem}
\label{t1}
    Assume that Condition (A1)-(A4) in  \cite{dette2022detecting} hold. Let $\hat q^*_{1-\alpha}$ denote the $(1-\alpha)$ quantile obtained by Algorithm \ref{Alg1} where $\hat I$ is calculated at level $\beta$. Further assume that $\sigma(t)^2>0$ for all $t \in [0,1]$ and that $\hat \sigma^2$ is a consistent estimator for $\sigma^2$. Then it holds, with $\hat \mu_i^\pm$ given as in \eqref{p3}, that
    \begin{align}
         \liminf_{n \to \infty} \p\Big(
         \cap_{i \in I}\big\{
  \text{\rm for all } t \in [0,1]:~ \hat \mu_i^-(t) \leq \mu_i(t) \leq \hat \mu_i^+(t) \big \}
 %        \cap_{i \in I}\Big\{\hat \mu_i^-(t) \leq \mu_i(t) \leq \hat \mu_i^+(t) \Big\}
         \Big)\geq 1-\alpha-\beta .
    \end{align}
    In the case that $\norm{\mu_{i+1}-\mu_i}_\infty \neq \Delta$ for $1 \leq i \leq m$ we can strengthen this result to to
    \begin{align}
         \lim_{n \to \infty} \p\Big(
         \cap_{i \in I}\big\{
  \text{\rm for all } t \in [0,1]:~ \hat \mu_i^-(t) \leq \mu_i(t) \leq \hat \mu_i^+(t) \big \}
%         \cap_{i \in I}\Big\{\hat \mu_i^-(t) \leq \mu_i(t) \leq \hat \mu_i^+(t) \Big\}
         \Big)= 1-\alpha.
    \end{align}
\end{theorem}

\begin{proof}
We provide only a sketch of the proof for the sake of brevity, focusing on the second statement. The first follows by similar arguments.  The proof proceeds in three steps, we first replace the change point estimators by their true values, then identify the relevant changes and finally employ weak invariance principles to establish the necessary asymptotic statements that yield validity of the bootstrap. In a bit more detail we 
\begin{enumerate}
    \item Use Theorem 4.1 from \cite{bastian2024multiplechangepointdetection} to replace the estimated change points by the true change points in the definitions of $\hat n_i$ and $\hat \mu_i$.
    \item Use Theorem 4.2 from \cite{bastian2024multiplechangepointdetection} to show that $\hat I=I$ with probability $1-o(1)$.
    \item Observe that multiplication by $\sigma$, taking component wise maxima and the supremum norm are continuous functions and that by Lemma \ref{pl1} (below) we may replace $\hat \sigma$ by $\sigma$. Then use Theorem 2.2 and 4.3  from \cite{Aue2020} in combination with the continuous mapping theorem to show that $T_n$ and finitely many copies $T^*_{n,i}$ of $T_n^*$ converge jointly in distribution, more precisely
    \begin{align}
        (T_n,T_{n}^{*, (1)},...,T_{n}^{*, (R)})\overset{d}{\rightarrow}(T,T^{(1)},...,T^{(R)}),
    \end{align}
    where $T_i$ are independent copies of $T$.     
\end{enumerate}
Combining these steps immediately yields validity of the bootstrap scheme.
    
\end{proof}

Now we provide the precise definition of the long run variance estimator $\hat \sigma^2$ and derive its consistency. Let
\begin{align}
\label{lrv}
    \hat \sigma^2=\sum_{l=-c}^c\hat{\sigma}^2_lK(l/c)
\end{align}
for some bandwidth parameter $c\rightarrow \infty$, kernel $K$ and where
\begin{align}
    \hat{\sigma}^2_{l} =
\begin{cases}
\frac{1}{n} \sum_{j=1}^{n-l} \big(X_j- \hat \mu^{(j)}\big)\big(X_{j+\ell} - \hat \mu^{(j)}\big), & l \geq 0, \\
\frac{1}{n} \sum_{j=1-l}^{n} \big(X_j- \hat \mu^{(j)}\big)\big(X_{j+\ell} - \hat \mu^{(j)}\big), & l < 0,
\end{cases}
\end{align}
are estimators for the lag variances and
$\hat \mu^{(j)}=\hat \mu_i$ for $\lfloor \hat s_i \rfloor \leq j < \lfloor \hat s_{i+1} \rfloor$. A symmetric function $K$ is called kernel whenever $K(0)=1, K(1)=0$ and $K(x)=0$ for $|x|>1$. \\
\begin{Lemma} 
\label{pl1}
    Assume that Condition (A1)-(A4) in  \cite{dette2022detecting} hold, then the long run variance estimator defined in \eqref{lrv} is consistent, i.e.  $ \norm{\hat \sigma^2-\sigma^2} =o_\p(1)$  whenever the bandwidth $c$ satisfies $c \rightarrow \infty$ and $c^3/n \rightarrow 0$.
\end{Lemma}
\begin{proof}
    We will show the result for mean zero processes $\epsilon_1,...,\epsilon_n$ satisfying Conditions (A1) to (A4) in \cite{dette2022detecting} and the estimator $ \hat \sigma^2=\sum_{l=-c}^c \hat \sigma_l^2K(l/c)$  where   
    \begin{align}
    \hat{\sigma}^2_l =
\begin{cases}
\frac{1}{n} \sum_{j=1}^{n-l} \epsilon_j\epsilon_{j+l}, & l \geq 0, \\
\frac{1}{n} \sum_{j=1-l}^{n}  \epsilon_j\epsilon_{j+l}, & l < 0,
\end{cases}
\end{align}
and we define $\sigma_l^2=\mathbb{E}[\epsilon_j\epsilon_{j+l}]$. The general result then follows by straightforward approximation arguments. Observe the following decomposition:
\begin{align*}
    \hat \sigma^2 -\sigma^2 &= - \sum_{|j|>c}\sigma_j^2+\sum_{j=-c}^c  \big ( \hat \sigma_j^2- \sigma_j^2  \big ) \\
    &- \sum_{j=-c}^c (\hat \sigma_j^2-\sigma^2_j)(1-K(j/c)) - \sum_{j=-c}^c\sigma_j^2(1-K(j/c))
\end{align*}
%\begin{align*}
%    \hat \sigma^2 -\sigma^2 &= \sum_{|j|>c}\sigma_j^2+\sum_{j=-c}^c\hat \sigma_j^2-\sigma_j^2 \\
%    &- \sum_{j=-c}^c (\hat \sigma_j^2-\sigma^2_j)(1-K(j/c)) + \sum_{j=-c}^c\sigma_j^2(1-K(j/c))
% \end{align*}
By Lemma A from \cite{Yoshi1978} we have that   $\sup_{t\in [0,1]} \sigma_j^2(t) \lesssim  \sqrt{\phi(j)}$ which is a summable sequence by (A4). 
Therefore the  first summand is of order $o(1)$. The fourth summand converges to $0$ by Lebesgue's Theorem. Finally, 
Arguments similar to those given  on page 18 of the supplement of \cite{Aue2020} and a union bound argument yield that    
$$
\max_{-c \leq j \leq c}\norm{\hat \sigma^2_j-\sigma^2_j}_\infty=O_\p(\sqrt{c/n}).
$$
This  implies that the second and third summand are of order $O_\p(\sqrt{c^3/n})$. 
\end{proof}

\section{Application to sports data}
%The parameter choices for \cref{Alg1} as well as for the BINSEG$(1,n,\xi_n)$ algorithm used therein are as in \cite{bastian2024multiplechangepointdetection}. The number of bootstrap repetitions in each case is set to $R = 100.$ 

\noindent
In this section, we explore the application of the methodology presented above to biomechanical knee angle data collected under a fatigue protocol obtained from the collaboration project \textit{Sports, Data, and Interaction\footnote{\url{http://www.sports-data-interaction.com/} }}. The data examples presented in this work are collected via marker-based optical motion capture systems placed in a laboratory while the athlete ran on a treadmill. 
Following relevant change detection to demarcate changes corresponding to fatigue we now study the range of motion of the athlete by means of the confidence bands we introduced. For a more detailed discussion of this dataset and the literature on fatigue-related change detection, we refer the interested reader to \cite{bastian2024multiplechangepointdetection} and \cite{basu2023fatiguedetectionsequentialtesting}. In all of the following data examples, Band $0$ corresponds to the confidence bands around the mean curve $\mu_0 $ i.e, the mean in the segments $[0, s_1]$, Band $1$ corresponds to mean $\mu_1$, i.e. to the segment $[s_{i_1}, s_{i_1+1}]$ and so on. 

\subsection{I: Confidence bands for knee angles}
\begin{figure}[H]
\centering
    \includegraphics[width=0.75\linewidth, height= 5.4cm]{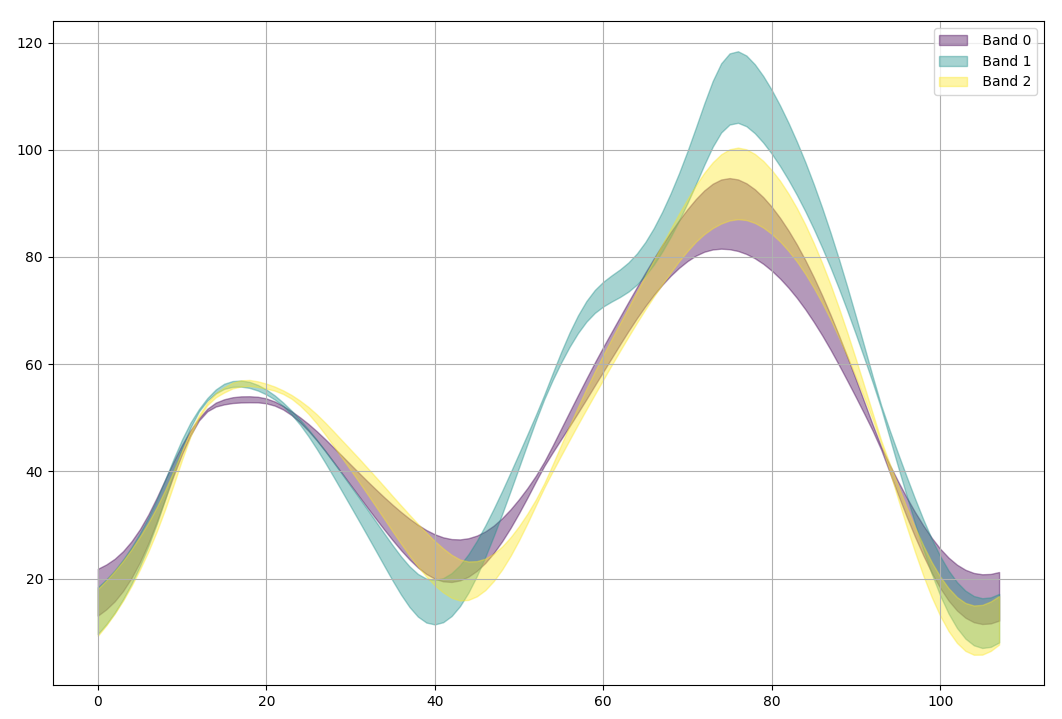}
    \caption{Confidence bands pre- and post- each change for right knee angles from runner A. Total sample size $n = 1830$ and three phases are detected for a relevant size $\Delta = 6.6.$}
    \label{fig:runner06CB}
    %runner A is 06 in dataset
\end{figure}

In Figure \ref{fig:runner06CB}, we present the confidence bands for a novice runner under fatiguing conditions. They have about 2 years of running experience and their dominant leg is the right leg. 
To determine the value of \( \Delta \) in a data-driven manner, we compute the average of the initial \( 5\% \) of the data, denoted as \( \hat{\mu}_{\text{initial}} \), and the average of the final \( 5\% \) of the data, denoted as \( \hat{\mu}_{\text{final}} \). Then, we set 
$\Delta = \frac{\lVert \hat{\mu}_{\text{final}} - \hat{\mu}_{\text{initial}} \rVert_\infty}{3}, $  similarly to \cite{bastian2024multiplechangepointdetection}. For the case of this runner,  we obtain $\Delta = 6.6.$ 
Band $0$ corresponds to the rest phase, Band $1$ corresponds to the pre-fatigue phase and Band $2$ finally corresponds to the fatigue phase. Notice that the intermediate shift upwards in the second peak of knee angles is an initial adjustment made by the runner to endure fatigue. Band $2$ indicates that the final adjustment corresponds to lesser bending of the knee while the foot is in the air, which is clearly reflected by the large jump in the second peak of the bands. The lower and upper bands at each part of the cycle provide useful information on the range of motion of the knee movement through different segments of the run. We stress that the exceedance threshold $\Delta$ eliminates all segments of non-fatigue related deviations in the knee angles.

\subsection{II: Comparison novice and experienced runner}
\begin{figure}[H]
\hspace{-.4cm}
\includegraphics[width=1.03\linewidth, height =5.5cm]{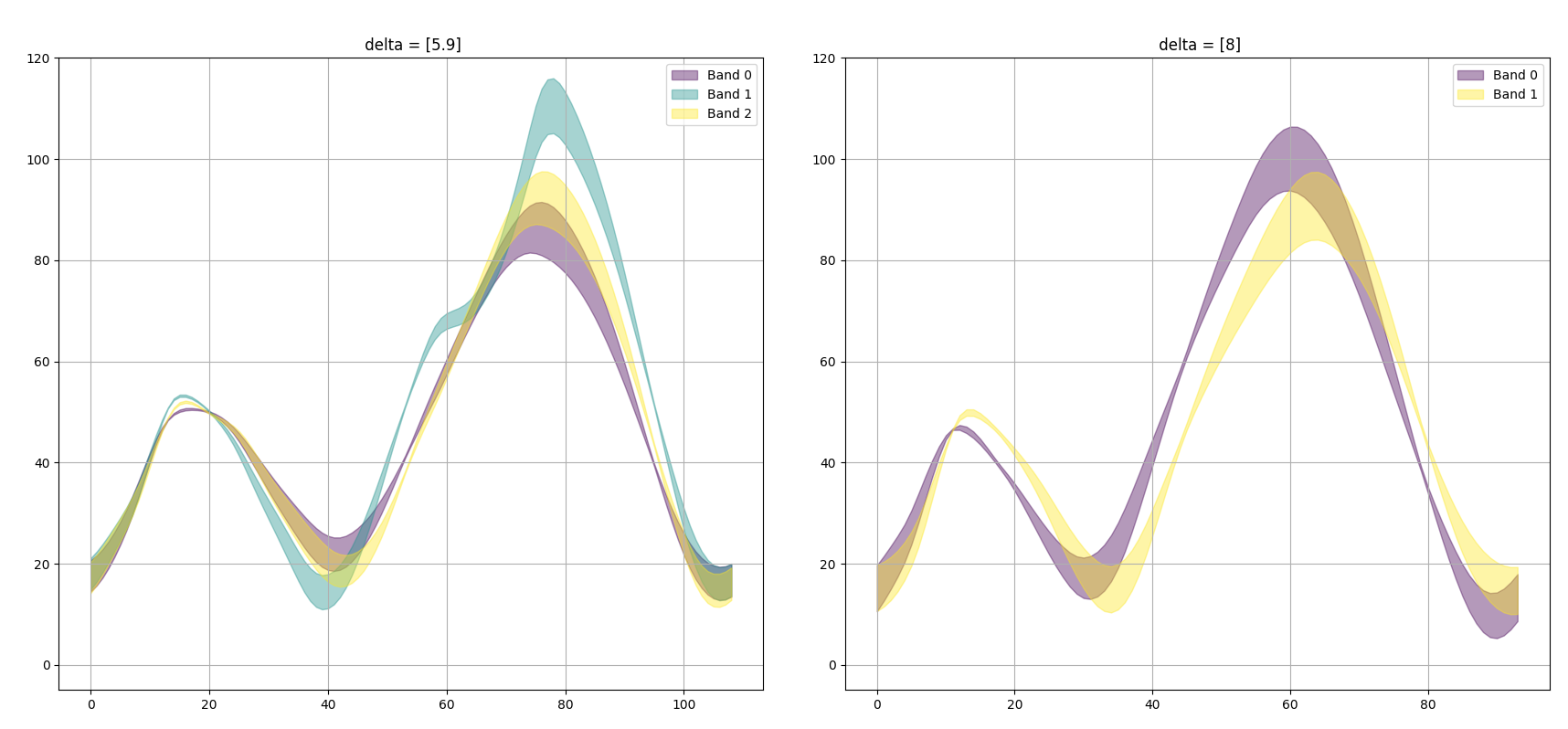}\caption{Confidence bands from two runners. \textbf{Left}: (Novice-) runner A  left knee angles for relevant size $\Delta = 5.9$ and \textbf{Right}: (experienced-) runner B left knee angles for relevant size $\Delta = 8$. }\label{fig:compare0609}
\end{figure}

In Figure \ref{fig:compare0609} we compare two runners with different experience levels. On the left we have the same runner A as in Figure \ref{fig:runner06CB} while on the right we have an experienced runner B with about 11 years of running experience and a dominant left leg.  This time we consider the left leg for both runners, notice that runner B has only one relevant change of size $\Delta = 8$. This may be due to B being well trained, i.e. their body is used to fatiguing conditions and therefore needs to make fewer adjustments to endure accumulated fatigue. However, just as for the less experienced runner A, we can see that in the fatigue phase there is reduced bending of the knee at the second peak. This behaviour is consistent across runners even beyond the two cases we present here and is in accordance with prior studies in the biomechanical gait literature, see for example \cite{zandbergen2023effects}, where this is attributed to a protection mechanism of the runner.    \\

% \begin{table}[]
%     \centering
%     \begin{tabular}{|c|c|c|c|c|c|}
%     \hline
%           &Runner 01& Runner 04 & Runner 06& Runner 07 & Runner 09 \\
%          \hline 
%           &0.585& 0.676&0.50307&0.72&0.6482\\
%           \hline
%           Years & 8 & 3& 2& 6& 11\\
%           \hline
%           per week& 15-20km & 20km& 10-20km& 10km&55-60km\\
%           \hline
%     \end{tabular}
%     \caption{Width of confidence bands for runners.}
%     \label{tab:my_label}
% \end{table}

\noindent
\textbf{Acknowledgments: } We thank R. van Middelaar and A. Balasubramaniam of the University of Twente., the Netherlands for the immense amount of work they put in providing high-quality data from running athletes.   The authors gratefully acknowledge the support by the Deutsche Forschungsgemeinschaft (DFG), project number 511905296 titled: \textit{Modeling functional time series with dynamic factor structures and points of impact} and  
 TRR 391 \textit{Spatio-temporal Statistics for the Transition of Energy and Transport}, project number 520388526.

\bibliographystyle{apalike}
\bibliography{bibliography}     

\begin{thebibliography}{}

\bibitem[Bastian et~al., 2024]{bastian2024multiplechangepointdetection}
Bastian, P., Basu, R., and Dette, H. (2024).
\newblock {Multiple change point detection in functional data with applications to biomechanical fatigue data}.
\newblock {\em The Annals of Applied Statistics}, 18(4):3109 -- 3129.

\bibitem[Basu and Proksch, 2023]{basu2023fatiguedetectionsequentialtesting}
Basu, R. and Proksch, K. (2023).
\newblock Fatigue detection via sequential testing of biomechanical data using martingale statistic. https://arxiv.org/abs/2306.01566.

\bibitem[Dette and Kokot, 2022]{dette2022detecting}
Dette, H. and Kokot, K. (2022).
\newblock Detecting relevant differences in the covariance operators of functional time series: a sup-norm approach.
\newblock {\em Annals of the Institute of Statistical Mathematics}, 74:195--231.

\bibitem[Dette et~al., 2020]{Aue2020}
Dette, H., Kokot, K., and Aue, A. (2020).
\newblock {Functional data analysis in the Banach space of continuous functions}.
\newblock {\em The Annals of Statistics}, 48(2):1168 -- 1192.

\bibitem[Yoshihara, 1978]{Yoshi1978}
Yoshihara, K.-i. (1978).
\newblock {Moment inequalities for mixing sequences}.
\newblock {\em Kodai Mathematical Journal}, 1(2):316 -- 328.

\bibitem[Zandbergen et~al., 2023]{zandbergen2023effects}
Zandbergen, M.~A., Marotta, L., Bulthuis, R., Buurke, J.~H., Veltink, P.~H., and Reenalda, J. (2023).
\newblock Effects of level running-induced fatigue on running kinematics: A systematic review and meta-analysis.
\newblock {\em Gait \& posture}, 99:60--75.

\end{thebibliography}
\newpage
\end{document}